\documentclass[11pt]{article}
\usepackage{times}
\usepackage{xspace}
\usepackage{tikz}
\usepackage{graphicx}

\newtheorem{definition}{Definition}
\newtheorem{theorem}{Theorem}

\newtheorem{claim}[theorem]{Claim}

\newtheorem{corollary}[theorem]{Corollary}

\newcommand{\sq}{\hbox{\rlap{$\sqcap$}$\sqcup$}}
\newcommand{\qed}{\hspace*{\fill}\sq}
\newenvironment{proof}{\noindent {\bf Proof.}\ }{\qed\par\vskip 4mm\par}

\newcommand{\PPP}{\textsf{PPP}\xspace}
\newcommand{\SPP}{\textsf{SPP}\xspace}
\newcommand{\MPP}{\textsf{MPP}\xspace}
\newcommand{\LVP}{\textsf{LVP}\xspace}

\newcommand{\removed}[1]{}

\title{The Dynamics of Probabilistic Population Protocols 
\thanks{This work has been partially supported by the ICT Programme of the European Union under contract number ICT-2008-215270 (\textsf{FRONTS}).} }

\author{Ioannis Chatzigiannakis
\thanks{Contact author} \\
R.A. Computer Technology Institute \\
and CEID, University of Patras\\
26500 Greece\\
\texttt{ichatz@cti.gr} \\
\and
Paul G. Spirakis\\
R.A. Computer Technology Institute \\
and CEID, University of Patras\\
26500 Greece\\
\texttt{spirakis@cti.gr}
}

\date{}

\begin{document}
\maketitle


\begin{abstract}
We study here the dynamics (and stability) of Probabilistic Population Protocols, via the differential equations approach. We provide a quite general model and we show that it includes the model of Angluin et. al. \cite{AADFP04}, in the case of very large populations. For the general model we give a sufficient condition for stability that can be checked in polynomial time. We also study two interesting subcases: (a) protocols whose specifications (in our terms) are configuration independent. We show that they are always stable and that their eventual subpopulation percentages are actually a Markov Chain stationary distribution. (b) protocols that have dynamics resembling virus spread. We show that their dynamics are actually similar to the well-known Replicator Dynamics of Evolutionary Games. We also provide a sufficient condition for stability in this case.
\end{abstract}

\section{Introduction}
\label{sec:Intro}

In the near future, it is reasonable to expect that new types of systems will appear, designed or emerged, of massive scale, expansive and permeating their environment, of very heterogeneous nature, and operating in a constantly changing networked environment. Such systems are expected to operate even beyond the complete understanding and control of their designers, developers, and users. Although they will be perpetually adapting to a constantly changing environment, they will have to meet their clearly-defined objectives and provide guarantees about certain aspects of their own behavior.

We expect that most such systems will have the form of a very large society of networked artefacts. Each such artefact will be unimpressive: small, with limited sensing, signal processing, and communication capabilities, and usually of limited energy. Yet by cooperation, they will be organized in large societies to accomplish tasks that are difficult or beyond the capabilities of todays conventional centralized systems. These systems or societies should have particular ways to achieve an appropriate level of organization and integration. This organization should be achieved seamlessly and with appropriate levels of flexibility, in order to be able to achieve their global goals and objectives. 

Angluin et al. \cite{AADFP04,AADFP06} introduced the notion of a computation by a population protocol to model such distributed systems in which individual agents are extremely limited and can be represented as finite state machines. In their model, finite-state, and complex behavior of the system as a whole emerges from the rules governing pairwise interaction of the agents. The computation is carried out by a collection of agents, each of which receives a piece of the input. These agents move around and information can be exchanged between two agents whenever they come into contact with each other. The goal is to ensure that every agent can eventually output the value that is to be computed (assuming a fairness condition on the sequence of interactions that occur). 

In \cite{AADFP04} they also proposed a natural probabilistic variation of the standard population protocol model, in which finite-state agents interact in pairs under the control of an adversary scheduler. In this variant, interactions that occurs between pairs of agents are chosen uniformly at random. We call the protocols of \cite{AADFP04} by the term ``Probabilistic Population Protocols'' (\PPP). In \cite{AAE06} they presented fast algorithms for performing computations in this variation and showed how to use the notion of a leader in order to efficiently compute semilinear predicates and in order to simulate efficiently LOGSPACE Turing Machines. \cite{DF01} studied the acquisition and propagation of knowledge in the probabilistic model of random interactions between all paris in a population (conjugating automata). 

In this work we characterize the dynamics of population protocols by examining the rate of growth of the states of the agents as the protocol evolves. We imagine here a continuoum of agents. By the law of large numbers, one can model the undelrying aggregate stochastic process as a deterministic flow system. Our main proposal here is to exploit the powerful tools of continuous nonlinear dynamics in order to examine questions (such as stability) of such protocols. 

We first provide a very general model for population protocol continuous dynamics. This model (Switching Population Protocols -- \SPP) includes the probabilistic population protocols (\PPP) of \cite{AADFP04} as a special case, when the population is infinite and the time is continuous.

We show a sufficient condition for stability of \SPP that can be checked in \textit{polynomial time}. We also examine two subclasses of \SPP:
\begin{itemize}

\item The \textit{Markovian Population Protocols} (\MPP). In these protocols, their \textit{specifications are configuration independent}. In this very practical case, we show that \MPP are \textit{always stable} and their \textit{unique} population mix at stability is exacly the steady-state distribution of a \textit{Markov Chain}.

\item The \textit{Linear Viral Protocols} (\LVP). They are probabilistic protocols motivated by the ``random pairing'' of \cite{AADFP04}. However, agents review their current state at a higher rate when they have weak ``immunity''. We view this as a general model for the dynamics of viruses spread in the population. We show that \LVP are equivalent to the well-known ``Replicator Dynamics'' of Evolutionary Game Theory, and thus to the well-known Lotka-Volterra dynamics. We also give a sufficient condition for stability of \LVP, based on Potentials.

\end{itemize}

\section{The General Model (Switching Probabilistic Protocols -- \SPP)}

The network is modeled as a complete graph $G$ where vertices represent nodes and edges represent communication links between nodes. We use the letter $n$ to denote $|V|$, the number of nodes in the network. Each node is capable of executing an ``agent'' (or process) which consists of the following components:

\begin{itemize}

\item $K$, a finite set of states. We use the letter $k$ to denote $|K|$.

\item $X$, a nonempty subset of $K$, known as the inital states or start states.

\end{itemize}

We consider a large population of $n$ agents. Let $q \in K$ be a state of the agent and let $n_q$ the number of agents that are on the given state $p$. Then the total population size is $n=\sum_{i=1}^{k} n_i$. The proportion of agents that are at state $q$ is $x_q = \frac{n_q}{n}$. We call $x_q$ the \textit{density} of $q$. In the sequel $q=q_i$, where $i \in \{1,2,\ldots,k\}$.

A state assignment of a system is defined to be an assignment of a state to each agent in the system. A \textit{configuration} $C$ is a map from the population to states, giving the current state of every agent. The population state then, at time $t$, can be described via a vector $\vec{x}(t) = \left(x_1(t), \ldots, x_k(t)\right)$. Here $x_i(t) = \frac{n_i}{n}$, $i=1 \ldots k$.

In the sequel we assume that $n \rightarrow \infty$. We are interested, thus, in the evolution of $\vec{x}(t)$ as time goes on. We use a different model (compared to \cite{AADFP04}) for describing a protocol $P$. We imagine that all agents in the population are infinitely lived and that they interact forever. Each agent sticks to some state in $K$ for some time interval, and now and then \textit{reviews} her state. This depends on $\vec{x}(t)$ and may result to a change of state of the agent. Based on this concept, a \textit{switching population protocol} consists of the following two basic elements (specifications):

\begin{enumerate}

\item A specification of the \textit{time rate} at which agents in the population review their state. This rate may depend on the current, ``local'', performance of the agent's state and also on the configuration $\vec{x}(t)$.

\item A specification of the \textit{switching probabilities} of a reviewing agent. The probability that an agent, currently in state $q_i$ at a review time, will \textit{switch} to state $q_j$ is in general a function $p_{ij}\left(\vec{x}(t)\right)$, where $p_i\left(\vec{x}\right) = \left(p_{i1}\left(\vec{x}\right), \ldots, p_{ik}\left(\vec{x}\right) \right)$ is the resulting distribution over the set $K$ of states in the protocol.

\end{enumerate}

In a large, finite, population $n$, we assume that the review times of an agent are the ``birth times'' of a Poisson process of rate $\lambda_i\left(\vec{x} \right)$. At each such time, the agent $i$ selects a new state according to $p_i\left(\vec{x}\right)$. We assume that all such Poisson processes are independent. Then, the aggregate of review times in the sub-population of agents in state $q_i$ is itself a Poisson process of birth rate $x_i\lambda_i\left(\vec{x}\right)$. As in the probabilistic model of \cite{AADFP04} we assume that state switches are independent random variables accross agents. Then, the rate of the (aggregate) Poisson process of switches from state $q_i$ to state $q_j$ in the whole population is just 
$x_i(t) \lambda_i\left(\vec{x}(t)\right) p_{ij}\left(\vec{x}(t)\right)$.

When $n\rightarrow\infty$, we can model the aggregate stochastic processes as deterministic flows (see, e.g., \cite{MU05,NW}). The outflow from state $q_i$ is $\sum_{j\ne i} x_j \lambda_j\left(\vec{x}\right)p_{ij}\left(\vec{x}\right)$. Then, the rate of change of $x_i(t)$ (i.e. $\frac{dx_i(t)}{dt}$ or $\dot{x_i}(t)$) is just
\begin{equation}\label{equ:1}
\dot{x_i}\, = \, \sum_{j \in K} x_j p_{ji} \left(\vec{x}\right) \lambda_j \left(\vec{x} \right) \, - \, \lambda_i \left(\vec{x}\right) x_i
\end{equation}
for $i = 1, \ldots, k$.

We assume here that both $\lambda_i\left(\vec{x}\right)$ and $p_{ij}\left(\vec{x}\right)$ are Lipschitz continuous functions in an open domain $\Sigma$ containing the simplex $\Delta$ where
$$
\Delta \, = \, \left\{ \left(x_i, \ldots, x_k \right) : \sum_{i=1}^{K} x_i = 1 \, , \qquad  x_i \ge 0 \ , \  \forall i \right\}
$$

By the theorem of Picard-Linderl\"of (see, e.g., \cite{HS74} for a proof), Eq.~\ref{equ:1} has a \textit{unique} solution for any initial state $\vec{x}(0)$ in $\Delta$ and such a solution trajectory $\vec{x}(t)$ is \textit{continuous} and never leaves $\Delta$.

\subsection{\SPP includes the probabilistic population protocols}

We now show that our model of Switching Probabilistic Protocols (\SPP) is more general than the model of \cite{AADFP04} in the sense that it can be used to define the Probabilistic Population Protocols (\PPP). We do this by showing the following:
\begin{theorem}{\rm
The continuous time dynamics of \PPP (when $n \rightarrow \infty$) are a special case of the dynamics of \SPP.
}\end{theorem}

\begin{proof}
According to \cite{AADFP04}, the discrete-time dynamics of a Probabilistic Population Protocol (\PPP) are given by a finite set of rules, $R$ of the form
$$
(p,q) \mapsto (p',q')
$$
where $p,q,p',q' \in K$ ($K=\{q_1, \ldots, q_k\}$) together with a set $A$ of $n$ agents and an (irreflexive) relation $E \subseteq A \times A$. 

Intruitively, a $(u,v) \in E$ means that $u,v$ are able to interact. \cite{AADFP04} assumes further that $E$ consists of all ordered pairs of distinct elements from $A$.

A \textit{population configuration} in \cite{AADFP04} is a mapping $C: A \mapsto K$ ($K$ is the set of states). Let $C$ and $C'$ be population configurations, and $u,v$ be two distinct agents. \cite{AADFP04} says that $C$ can go to $C'$ in one discrete step (denoted $C \stackrel{e}{\mapsto} C'$) via an \textit{encounter} $e=(u,v)$ if
$$
\left( C(u), C(v) \right) \mapsto \left(C'(u), C'(v)\right)
$$
is a rule in $R$. This means that the state $C(u)$ of $u$ switches to $C'(u)$ and also $C(v)$ switches to $C'(v)$.

The execution of the system is defined to be a sequence $C_0, C_1, C_2, \ldots$ of configurations (where $C_0$ is the initial configuration) such that for each $i$, $C_i \mapsto C_{i+1}$. An execution is fair if for any $C_i$ and $C_j$, such that $C_i \mapsto C_j$ and $C_i$ occurs infinitely often in the execution, $C_j$ also occurs infinitely often in the execution.

In the probabilistic version of the above, \cite{AADFP04} further states that $e$ (the ordered pair to interact) is chosen at random, independently and uniformly from all ordered pairs corresponding to edges $e$ in $A \times A$ (\cite{AADFP04} calls it the model of Conjugating Automata, inspired also by \cite{DF01}).

Let us now assume that $n \rightarrow \infty$ and let $x_i = \lim_{n \rightarrow \infty} \frac{n_i}{n}$ be the population fraction at state $q_i \in K$ at a particular configuration $C$, at time $t$. Consider the rule $\rho$ in $R$
$$
(q_r, q_m) \mapsto (q_i, q_j)
$$
%
%
%
%
Without loss of generality we assume in the sequel that $r \ne m$ and $i \ne j$ in such rules $\rho$ in $R$.
By the uniformity and randomness, the probability that such an $e$, that follows from rule $\rho$, is selected (as the encounter), is just $x_r(t) x_m(t)$. Let $A_i$ be the set of all $(r,m)$ that are the left part of a rule $\rho$:
\begin{eqnarray*}
(q_r, q_m) &\mapsto& (q_i, q_j) \\
\mbox{or} \qquad (q_r, q_m) &\mapsto& (q_j, q_i)
\end{eqnarray*}
Let $B_i$ be the set of $(r,m)$ that are the left part of a rule $\rho'$:
$$
(q_r, q_m) \mapsto (q_{r'}, q_{m'})
$$
with $r=i$ or $m=i$. Without loss of generality let $r=i$ in $\rho'$. By considering a small interval $\Delta t$ and taking limits as $\Delta t \rightarrow 0$, due to fairness we get $\forall i$:
\begin{equation}\label{equ:1star}
\dot{x_i} \, = \, \sum_{(r,m) \in A_i} x_r(t) x_m(t) \, - \, x_i(t) \sum_{(i,m) \in B_i} x_m(t)
\end{equation}
The above set of equations describe the continuous dynamics of \PPP.

Now, consider our \SPP dynamics and Eq.~\ref{equ:1}.
Set $\lambda_i\left(\vec{x}\right) = \sum x_m(t)$, with $m$ ranging over all rules
$$
(q_r, q_m) \mapsto (q_{r'}, q_{m'})
$$
with $r=i$, and all rules
$$
(q_m, q_r) \mapsto (q_{r'}, q_{m'})
$$
with $r=i$ (i.e., over all rules in $B_i$).

Also, set $p_{mi} = p_{ri}$ = 0, if $r,m$ do not belong in any tuple of $A_i$.

Finally set
$$
p_{ri} = \frac{1}{\lambda_r} \sum_{m \in C(r,i)} x_m(t)
$$
where $C(r,i)$ is the set of indices $m$ in the second argument of the left part of rules in $A_i$
(i.e. $(q_r, q_m) \mapsto (q_{r'}, q_{m'})$ with $r'=i$ or $m'=i$).

Then our system of Eq.~\ref{equ:1} (the \SPP dynamics) becomes the system of Eq.~\ref{equ:2} (the \PPP dynamics). Thus the \PPP dynamics are a special case of the \SPP dynamics in the continuous time setting.
\end{proof}

Here is an example of the reduction described above. Let the rules $R$ in \PPP be
\begin{eqnarray*}
(q_1, q_2) &\mapsto& (q_3, q_2) \\
(q_3, q_1) &\mapsto& (q_1, q_2) \\
(q_2, q_3) &\mapsto& (q_2, q_1)
\end{eqnarray*}
This gives the continuous \PPP dynamics:
\begin{eqnarray*}
\dot{x_1} &=& x_1x_3 \, + \, x_2x_3 \, - \, x_1 \left(x_2 + x_3 \right) \\
\dot{x_2} &=& x_1x_3 \, + \, x_1x_2 \, + \, x_2x_3 \, - \, x_2 \left(x_1 + x_3 \right) \\
\dot{x_3} &=& x_1x_2 \, - \, x_3 \left(x_1 + x_2 \right) 
\end{eqnarray*}
We then set
\begin{eqnarray*}
\lambda_1 &=& x_2 \, + \, x_3 \\
\lambda_2 &=& x_1 \, + \, x_3 \\
\lambda_3 &=& x_1 \, + \, x_2 
\end{eqnarray*}
and
$$
\begin{array}{l l l}
p_{21} = \frac{x_3}{x_1 + x_3} \qquad \qquad & p_{11} = \frac{x_3}{x_2 + x_3} \qquad \qquad & p_{31} = 0\\
p_{12} = \frac{x_3}{x_2 + x_3} & p_{22} = \frac{x_1}{x_1 + x_3} & p_{32} = \frac{x_2}{x_1 + x_2} \\
p_{13} = \frac{x_2}{x_2 + x_3} & p_{23} = p_{33} = 0 & 
\end{array}
$$
and this results in our \SPP dynamics, namely:
\begin{eqnarray*}
\dot{x_1} &=& x_1 \lambda_1 p_{11} \, + \, x_2 \lambda_2 p_{21} \, + \, x_3 \lambda_3 p_{31} \, - \, x_1 \lambda_1 \\
\dot{x_2} &=& x_1 \lambda_1 p_{12} \, + \, x_2 \lambda_2 p_{22} \, + \, x_3 \lambda_3 p_{32} \, - \, x_2 \lambda_2 \\
\dot{x_3} &=& x_1 \lambda_1 p_{13} \, + \, x_2 \lambda_2 p_{23} \, + \, x_3 \lambda_3 p_{33} \, - \, x_3 \lambda_3
\end{eqnarray*}

\section{Stability of nonlinear dynamic systems: a sufficient condition for decidability.}

Let us consider a dynamic system 
$$
\dot{x_i} \, = \, f_i \left(\vec{x}\right) \, , \qquad i= 1, \ldots, k
$$
that is, in fact, more general than Eq.~\ref{equ:1}.

\begin{definition}[Fixed Points]{\rm
Let $\vec{x}^*$ be a solution of the system $\left\{ f_i \left(\vec{x}^* \right) = 0, \, i=1,\ldots, k \right\}$ which we call a \textit{fixed point} of the system.
}\end{definition}

By making a Taylor expansion around $\vec{x}^*$ we obtain a linear approximation to the dynamics:
$$
\dot{x_i} \, = \, \sum \left(x_j - x_j^* \right) \frac{df_i}{dx_j} \left(\vec{x}^* \right)
$$
Setting $\xi_i = x_i - x_i^*$ we get
$$
\dot{\xi_i} \, = \, \sum \xi_j \frac{df_i}{dx_j} \left(\vec{x}^* \right)
$$
which is a Linear System with a fixed point at the origin, i.e., $\dot{\xi} = L \xi$ where the matrix $L$ has \textit{constant} components $L_{ij} = \frac{df_i}{dx_j} \left(\vec{x}^* \right)$. $L$ is called the Jacobian Matrix. Then, by the theorem of \cite{H60} we have
\begin{corollary}{\rm
If the fixed point $\vec{x}^*$ is \textit{hyperbolic} (i.e., all eigenvalues of $L^*$ have a non-zero real part) then the topology of the dynamics of the nonlinear system around $\vec{x}^*$ is the same as the topology of a $\vec{x}^*$ in the Linear system.
}\end{corollary}
In fact, let each eigenvalue of $L$ be $\phi = \texttt{a} + i \omega$.
\begin{corollary}{\rm
Let $\texttt{a} \ne 0$, $\forall \phi$ eigenvalues of $L$. Then
\begin{itemize}

\item[(a)] If $\texttt{a} < 0$, $\forall \phi$ then $\vec{x}(t)$ approaches the fixed point $\vec{x}^*$ as $t \rightarrow \infty$.

\item[(b)] If there exists a $\phi$ with $\texttt{a} > 0$ then $\vec{x}(t)$ \textit{diverges} from the fixed point $\vec{x}^*$ along the direction of the corresponding eigenvector. That is, the fixed point $\vec{x}^*$ is unstable.

\end{itemize}
}\end{corollary}
Thus we get our main result of the system:
\begin{theorem}{\rm
If all fixed points $\vec{x}^*$ of our population dynamics of Eq.~\ref{equ:1} are hyperbolic, then we \textit{can decide stability} of the population protocol, around $x^*$, in \textit{polynomial time} in the description of the protocol.
}\end{theorem}

\begin{corollary}\label{cor:1}{\rm
If all fixed points of \PPP are hyperbolic, then the stability of \PPP can be decided in polynomial time.
}\end{corollary}

%
%
%
%
%
%

\section{Switching Population Protocols with specifications independent of the configuration}

We now consider the special case of Eq.~\ref{equ:1} where $\lambda_i \left( \vec{x} \right) = \lambda_i \> \forall i$ and where $p_{ij} \left( \vec{x} \right) = p_{ij}$ (specifications independent of the configuration $\vec{x}(t)$). Then the basic system of Eq.~\ref{equ:1} of the dynamics of the population becomes:
\begin{equation}\label{equ:2}
\dot{x_i} \, = \, \sum_{j \in K} x_j \lambda_j p_{ji} \, - \, \lambda_i x_i \qquad \qquad i = 1 \ldots k
\end{equation}
We call such protocols by the term ``Markovian Population Protocols'' (\MPP).

Let $q_{ij} = \lambda_i p_{ij}$ for all $i,j$, when $i \ne j$ and when $j = i$ let $q_{ii} = \lambda_i (p_{ii} - 1)$. 
Then Eq.~\ref{equ:2} in fact becomes
\begin{equation}\label{equ:3}
\frac{dx_i(t)}{dt} \, = \, q_{ii} x_i(t) \, + \, \sum_{j \ne i} q_{ki} x_k(t)
\end{equation}
Note that $\sum_{i \in K} x_i (t) = 1$.  But this is, in fact, the basic equation of the limiting-state probabilities of a Markov Chain of $k$ states with $q_{ij}$ being the (continuous time) rates of change (see, e.g., \cite{Klein}, pp.~53--55).

When all $\lambda_{ij}$, $i \ne j$ are non zero then the Markov Chain of Eq.~\ref{equ:3} is irreducible and homogeneous. Then the limits $\lim_{t \rightarrow \infty} x_i(t)$ always exist and are independent of the initial state. The limiting distribution is given \textit{uniquely} as the solution of the following equations:
$$
q_{jj} x_j \, + \, \sum_{k \ne j} q_{kj} x_k \, = \, 0
$$
So, we get our second major result:
\begin{theorem}[Markovian Population Protocols -- \MPP]\label{the:2}{\rm
Let the specifications $\left\{ \lambda_j, p_{ij} \right \}$ independent of the configuration $\vec{x}(t)$. 
Let also $\lambda_j p_{ij} \ne 0$, $\forall i,j$ where $i \ne j$. Then the Population Protocol is \textit{stable}. It always has a 
limiting \textit{unique} configuration $\left \{ x_i \ \ i=1 \ldots k \right \}$ independent of the initial configuration $\vec{x}(0)$, which is exactly the \textit{steady-state distribution} of an \textit{ergodic, homogeneous Markov Chain of $k$ states}.
}\end{theorem}

\section{A special case of Random pairing population protocols\\ (Linear Viral Protocols -- \LVP)}

Now, let us assume that all reviewing agents adopt the state of ``the first man they meet in the street''. This is clearly the case when the reviewing agent draws a pairing agent at random from the population (according to the uniform probability distribution across agents) and adopts the state of the so sampled agent. This is similar to the case of the protocols of \cite{AADFP04} where the rules are $(q_i, q_k) \mapsto (q_m, q_r)$ with $r,m \in \{i,j\}$. Formally then
$$
p_{ij} \left( \vec{x} \right) = x_j \qquad \quad \forall i,j \in K, \, \forall x(t)
$$
Now Eq.~\ref{equ:3} becomes
$$
\dot{x_i} \, = \, \sum_{j \in K} x_j x_i \lambda_j (x) \, - \, \lambda_i(x) x_i
$$
i.e.
\begin{equation}\label{equ:51}
\dot{x_i} \, = \, \left( \sum_{j \in K} x_j \lambda_j (x) - \lambda_i(x) \right) \cdot x_i
\end{equation}

We now propose a ``linear'' model in order to capture the immunity that an agent has against other agents in the population. We postulate that agents immunity depend on their states. One can imagine immunity to be a measure of the degree of protection of agents when they interact. So, when an agent in state $q_i$ interacts with an agent in state $q_j$ we measure the immunity of the $(q_i, q_j)$ pair by an integer $\texttt{a}_{ij}$ and we require here that $\texttt{a}_{ij} = \texttt{a}_{ji}$. It is then natural to assume that agents in state $q_i$ will wish to review their state more often when their immunity is low. In particular we assume here that any agent in state $q_i$ has a review rate $\lambda_i \left( \vec{x} \right)$ that is \textit{linearly decreasing} in the average immunity of the agent in state $q_i$. This is the simplest possible model. the formal definitions follow:
\begin{definition}[Immunity of a state]
Let $A=\{\texttt{a}_{ij}\}$ be a symmetric matrix of integers. The immunity of an agent in state $q_i$ is 
$t_i\left(\vec{x}\right) = \texttt{a}_{i1}x_1 + \ldots + \texttt{a}_{ik}x_k$.
\end{definition}
\begin{definition}[Average immunity of a population protocol, in a particular configuration]
Let $A$ be a symmetric matrix of integers. The average immunity of the population, in configuration $\left\{ x_i \right\}$, is:
$t\left(\vec{x}\right) = \sum_{i \in K} x_i t_i \left(\vec{x} \right)$.
\end{definition}

\begin{definition}[Linear Viral Protocols -- \LVP]
The Linear Viral Protocols are switching population protocols with review rates of agents
$$
\lambda_i\left(\vec{x} \right) = \gamma - \delta t_i \left( \vec{x} \right)
$$
where $\gamma, \delta \in \Re$, $\delta > 0$ and also $\gamma / \delta \ge t_i \left( \vec{x} \right)$, 
$\forall \vec{x} + \Delta$, $\forall i$.
\end{definition}
Now Eq.~\ref{equ:51} becomes
\begin{equation}\label{equ:52}
\dot{x_i} \, = \, \delta \left( t_i\left(\vec{x}\right) - t\left(\vec{x}\right) \right) x_i
\end{equation}
Note, now, that this equation is a constant rescaling of the popular ``replicator dynamics''  of Evolutionary Game Theory (see, e.g., \cite{W97}).

\begin{definition}
The general Lotka-Volterra equation for $k$ types of a population is of the form
$$
\dot{x_i} \, = \, x_i \left( r_i + \sum_{j=1}^k \texttt{a}_{ij} x_j \right) \qquad i=1\ldots k
$$
where $r_i$, $\texttt{a}_{ij}$ are constant.
\end{definition}
By the equivalence of the Replicator Dynamics with the Lotka-Volterra systems we then get:
\begin{theorem}{\rm
The dynamics of the linear viral protocols are \textit{equivalent} to the Lotka-Volterra dynamics.
}\end{theorem}
We can then give an alternative sufficient condition for the (asymptotic) stability of the Linear Viral Protocols.
\begin{theorem}\label{the:5}{\rm
Let $x^*$ be a fixed point of Eq.~\ref{equ:52}, i.e.,
$t_i \left(\vec{x}\right) = t \left(\vec{x}\right)$ is satisfied for $\vec{x} = \vec{x}^*$.
If $\sum_{i=1}^{k} x_i^* t_i \left( \vec{x} \right) > t\left(\vec{x} \right)$ for any $\vec{x}$ in a region around $\vec{x}^*$, then $\vec{x}^*$ is asymptotically stable.
}\end{theorem}
In order to prove our theorem, we first consider the relative entropy of $\vec{x}$ and $\vec{x}^*$ as
\begin{equation}\label{equ:proof}
E(x) = - \sum_{i=1}^{k} x_i^* \ln\left(\frac{x_i}{x_i^*}\right) 
\end{equation}
Clearly $E(x^*) = 0$.
Then we need to prove the following claim:
\begin{claim}\label{claim:1}
$E(x) \ge E(x^*)$, $\forall \vec{x}$
\end{claim}

\begin{proof}
From Jensen's inequality it folds:
$$
\exp\left( f(x) \right) \, \ge \, f(\exp x)
$$
where $\exp()$ is the expectation, $x$ a random variable and $f$ a convex function. Thus Eq.~\ref{equ:proof} becomes
$$
E(x) \ge - \ln\left( \sum_{i=1}^k x_i^* \frac{x_i}{x_i^*}\right) \ge -\ln\left(\sum_{i=1}^k x_i \right) = - \ln 1 = 0
$$
\end{proof}

\begin{proof}
Based on Claim~\ref{claim:1} we can prove Theorem~\ref{the:5} as follows:
\begin{eqnarray*}
\frac{dE\left(\vec{x}(t)\right)}{dt} 
&=& \sum_{i=1}^k\frac{dE}{dx_i} \dot{x_i} \\
&=& - \sum_{i=1}^k \frac{x_i^*}{x_i} \dot{x_i} \\
&=& - \sum_{i=1}^k \delta \left( t_i\left(\vec{x}\right) - t\left(\vec{x}\right) \right) x_i^* \qquad \mbox{(due to Eq.~\ref{equ:52})}\\
&=& - \delta \left[ \sum_{i=1}^k \vec{x}^* \left( t_i(x) - t\left(\vec{x}\right)  \right) \right]\\
&<& 0 \qquad \mbox{by assumption}
\end{eqnarray*}
Thus, in a region around $\vec{x}^*$, $\frac{dE}{dt} < 0$.
Then $E$ is a (strict) Lyapounov function (see, e.g., \cite{HS98}, pp. 18--19) and thus $\vec{x}^*$ is stable asymptotically.
\end{proof}

\section{Conclusions}

The \textit{population protocol} model of Angluin et. al. \cite{AADFP04} consists of a (large) population of finite-state \textit{agents} that interact in pairs. Each interaction updates the state of both participants according to a transition function based on the pair of the participants' previous states. A natural probabilistic model, proposed in \cite{AADFP04}, assumes each interaction to occur between a pair of agents chosen uniformly at random. We call the protocols of \cite{AADFP04} by the term ``Probabilistic Population Protocols'' (\PPP). \cite{DF01} studied the acquisition and propagation of knowledge in the probabilistic model of random interactions between all paris in a population (conjugating automata). Curiously, the differential equation approach for such protocols was not proposed till now.

We imagine here a continuoum of agents. By the law of large numbers, one can model the underlying aggregate stochastic process as a deterministic flow system. Our main proposal here is to exploit the powerful tools of continuous nonlinear dynamics in order to examine questions (such as stability) of such protocols. 

We have extended the class of \cite{AADFP04} by defining a general model of ``Switching Population Protocols'' (\SPP). We then examined stability for this general model and two important subclasses. Our main point is that one can study stability and population dynamics of protocols, via nonlinear differential equations that describe quite accurately the (discrete) population protocol dynamics when the population is very large. The ``differential equations'' approach was indicated in the past for the analysis of evolution of algorithms with Random Inputs, by \cite{MU05,NW}. Our approach provides a sufficient condition for stability of \PPP of \cite{AADFP04} that can be checked in polynomial time. It also gives a more general way to \textit{specify} population protocols, that reveals interesting classes.


\end{document}